\documentclass[conference]{IEEEtran}
\IEEEoverridecommandlockouts
\usepackage{etex}
 
\usepackage{comment} 
\usepackage{fullpage} 
\usepackage[english]{babel}
\usepackage[utf8]{inputenc}
\usepackage{dsfont}
\usepackage{amsthm}
\usepackage{amsmath,etoolbox}

\usepackage{float}
\usepackage[pdftex]{graphicx}
\usepackage{epstopdf}
\usepackage{graphicx}
\usepackage{cite}
\usepackage{amssymb}
\usepackage{mathtools}

\usepackage{multirow} 
\usepackage{multicol}
\usepackage{longtable}
\usepackage{tabu}
\usepackage{supertabular,booktabs}

\usepackage[ruled,vlined]{algorithm2e}
\makeatletter
\renewcommand{\@algocf@capt@plain}{above}
\makeatother

\usepackage{tikz}
\usetikzlibrary{automata,arrows,positioning,calc}

\usepackage[a4paper,top=2cm,bottom=2.5cm,left=1.5cm,right=1.5cm,marginparwidth=1.84cm]{geometry}

\usepackage[colorinlistoftodos]{todonotes}

\newtheorem{theorem}{Theorem}
\newtheorem{definition}{Definition}

\newtheorem{proposition}{Proposition}

\def\BibTeX{{\rm B\kern-.05em{\sc i\kern-.025em b}\kern-.08em
    T\kern-.1667em\lower.7ex\hbox{E}\kern-.125emX}}

\begin{document}

\title{Optimal Transmission Policies for Energy Harvesting Age of Information Systems with Battery Recovery
}
\author{\IEEEauthorblockN{Caglar Tunc\IEEEauthorrefmark{1},
		Shivendra Panwar\IEEEauthorrefmark{1} }
\IEEEauthorblockA{\IEEEauthorrefmark{1}Department of Electrical and Computer Engineering,
NYU Tandon School of Engineering\\
NY, USA,
Email: \{ct1909, sp1832\}@nyu.edu}\thanks{This work is supported by the NY State Center for Advanced Technology in Telecommunications (CATT), NYU Wireless and an Ernst Weber Fellowship.}
\vspace{-8mm}}

\maketitle
\begin{abstract}
	We consider an energy harvesting information update system where a sensor is allowed to choose a transmission mode for each transmission, where each mode consists of a transmission power-error pair. We also incorporate the battery phenomenon called \textit{battery recovery effect} where a battery replenishes the deliverable energy if kept idle after discharge. For an energy-limited age of information (AoI) system, this phenomenon gives rise to the interesting trade-off of recovering energy after transmissions, at the cost of increased AoI. Considering two metrics, namely peak-age hitting probability and average age as the worst-case and average performance indicators, respectively, we propose a framework that formulates the optimal transmission scheme selection problem as a Markov Decision Process (MDP). We show that the gains obtained by considering both battery dynamics and adjustable transmission power together are much higher than the sum gain achieved if they are considered separately. We also propose a simple methodology to optimize the system performance taking into account worst-case and average performances jointly.
\end{abstract}

\section{Introduction}
\vspace{-1mm}
\emph{Age of information} (AoI) has attracted wide attention in the literature as a useful concept to analyze the freshness of the data in information update systems \cite{kaul_2011,kaul_2012}. 
The concept of AoI focuses on the analysis of low-power and energy-limited devices, such as sensors, which can take advantage of energy harvesting technologies as a sustainable energy resource for transmissions \cite{paradiso_2005}. In energy harvesting-powered sensory and monitoring systems, the energy harvesting profile plays a significant role in the functionality of the system because of the limitations introduced by the randomness of energy sources. In such scenarios, analyzing AoI-based performance of the system is of particular interest due to the energy-delay trade-off \cite{yates_2015}. Energy harvesting information update systems has been studied in the AoI literature for various scenarios \cite{baknina_2018,farazi_2018,feng_2018,feng_2018_noisy,bacinoglu_2017,bacinoglu_2018,yates_2015}. References \cite{feng_2018,feng_2018_noisy,bacinoglu_2017} investigate the optimal scheduling policies to minimize average AoI under transmission errors whereas the same metric is used in \cite{bacinoglu_2018,farazi_2018,yates_2015} for error-free channel. The study in \cite{baknina_2018} aims to identify the achievable average message rate and AoI for an energy harvesting AoI system. In these studies, energy units are assumed to be arriving randomly at each time slot, either according to a simple Bernouili \cite{baknina_2018} or Poisson distribution \cite{farazi_2018,feng_2018,feng_2018_noisy,bacinoglu_2017,bacinoglu_2018,yates_2015}. On the other hand, existing literature on the design of wireless sensor networks emphasizes the impact of effectively modeling the energy harvesting process on the analysis of the functionality and lifetime of the network \cite{kansal_2003,jiang_2005}. We incorporate such a model in this paper.

There are a few studies in the context of AoI that investigate the throughput-error trade-offs for the source/channel coding schemes by adjusting the blocklength for different arrival and service rates, using queuing analytic methods \cite{zhong_2017,sac_2018,devassy_2018}. However, when there is no data buffer but the transmitter has limited transmission opportunities due to energy constraints, channel errors and transmission/retransmission schemes play a significant role in the performance of an AoI system \cite{ceran_2018}. In such scenarios, allowing the transmitter to adapt the transmission mode and exploiting the power-error trade-off is a promising approach to improve the performance of the AoI system, which we focus on in this paper.

In this paper, we analyze optimal transmission policies for an energy harvesting sensor, where the main objective is to optimize the system performance in terms of AoI-related metrics. When evaluating the performance, we focus both on the worst-case and average performances and consider \emph{peak-age hitting probability}, which is the probability that the age hits a predetermined threshold, and \emph{average AoI} as the performance metrics. We use a finite-state Markov model for the energy harvesting process, which quantifies the amount of harvested energy at different states and can be used to model realistic random energy profiles. 
In the MDP formulation, we also incorporate the battery dynamic called \emph{battery recovery effect}, which refers to the ability of a battery to self-replenish deliverable energy when left idle after a discharge, due to its chemical properties \cite{chiasserini_1999}. Impacts of battery receovery on battery lifetime and its implications on traffic shaping have been well studied in the sensor literature \cite{poellabauer_2004,chau_2010, fu_2018}. 
For an AoI system, this phenomenon gives rise to the interesting trade-off between keeping the battery idle for some time period to allow energy recovery at the cost of increasing AoI, and transmitting information updates as frequently as possible with less recovered energy. 
Moreover, using adaptive transmission schemes by changing the transmission power for such low-power devices mitigates battery limitations, which has not been addressed in the AoI literature to the best of our knowledge. Adjusting the transmission power might be even more crucial when coupled with the effects of the battery recovery effect.

Our main contributions in this paper can be summarized as follows:
\begin{itemize}
	\item We propose an MDP framework to find the optimal transmission policies for an AoI system that incorporates the battery recovery effect and multiple TX modes.
	\item We show with numerical examples that the system benefits from jointly exploiting the recovery effect and multiple TX modes, in terms of two performance metrics, namely average AoI and peak-age hitting probability.
	\item After demonstrating that an average-optimal policy may actually perform poorly in terms of the worst-case performance, we propose an effective strategy that takes into account both average and worst-case performances.
\end{itemize}
The paper is organized as follows. In Section II, we describe the energy harvesting sensor model and the MDP formulation of the system. We propose the optimal scheduling and transmission mode selection policies in Section III. After numerically investigating the system performance for various scenarios in Section IV, we conclude the paper.

\section{System Model}
We consider an energy harvesting sensor that generates information update packets and transmits them to a base station (BS) over a discrete-time channel with transmission errors, as illustrated in Fig.~\ref{fig:sensor_model}. We assume that there is no data buffer to accumulate packets and the transmission of an update packet is initialized immediately after being generated. For the transmission of an update packet, the sensor chooses a \emph{transmission (TX) mode} from set $\mathcal{M}$ which contains $M$ candidate schemes. The transmission mode determines the power level to be used for the transmission and the corresponding error probability, where lower (higher) TX power will result in a larger (smaller) error probability. TX power and error probability in TX mode $i\in \mathcal{M}$ are denoted by $P_i$ and $p_i$, respectively. Transmission of an information update packet is assumed to take one time slot. We denote the AoI at time $t$ by $A(t)$. At the beginning of a time slot, AoI is increased if there is no update packet received at that time slot. If there is an update packet received by the BS, AoI is updated to one, which is the transmission time of the packet. On the other hand, if the transmission is unsuccessful, then the AoI keeps increasing. We assume that the BS transmits instantaneous and error-free acknowledgement (ACK).
\begin{figure}[h!]
	\vspace{-2mm}
	\centering
	\includegraphics[width=7cm]{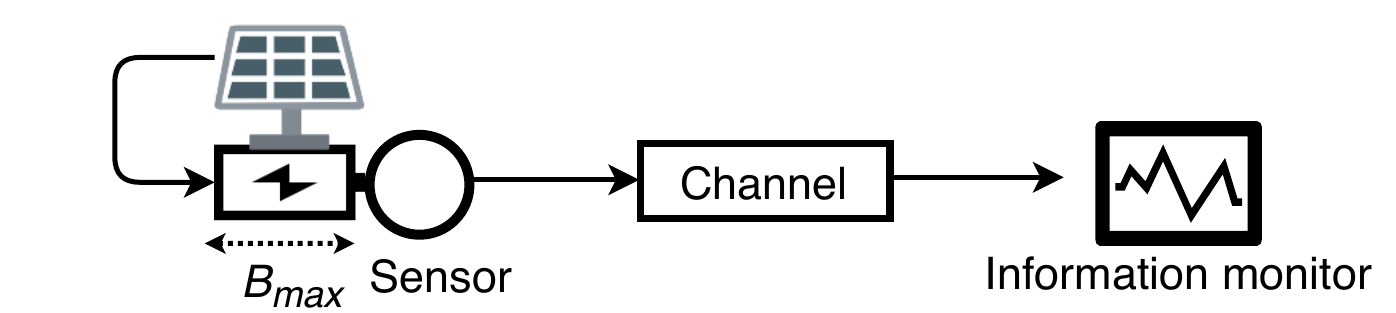}
	\caption{Energy harvesting sensor model.}
	\label{fig:sensor_model}
\end{figure}
\vspace{-3mm}

\subsection{Battery Recovery Effect}
In order to represent the battery recovery effect, we use a Markov model with $N_{rec}$ states, for which the recovery process takes place as follows. After a transmission with TX power $P_i$, $P_i/N_{rec}$ units of energy are recovered during an idle time slot with probability $p_{rec}\leq 1$ for at most $N_{rec}$ consecutive idle slots following the transmission. Note that this is a simplified version of the model in \cite{chiasserini_1999} but is still capable of capturing the random nature of energy recovery.



\subsection{Energy Harvesting Process}
We denote the capacity of the battery by $B_{max}$. We assume that the battery stores the harvested energy which is generated according to a discrete-time Markov chain with $N_H$ states. The energy harvester output power is fixed to $P_{H_i}$ in state $H_i$ for $i=1,...,N_H$, which means that the battery may store up to $P_{H_i}$ units of energy in a single time slot at harvester state $H_i$. The battery constraints $B(t)\leq B_{max}$ and $B(t)\geq0$ ensure that the battery level is always in the interval $[0,B_{max}]$, where $B(t)$ denotes the battery level at time $t$. Hence, if there is no transmission and $B(t)+P_{H_i}>B_{max}$, $B(t+1)$ is set to $B_{max}$ and the remaining harvested energy is lost. We denote the energy harvester state at time $t$ by $H(t)$, which changes in each time slot according to the transition matrix $\mathbf{P_H}=\{q_{ij}\}$ where $q_{ij}$ denotes the probability of transition from state $H_i$ to state $H_j$ at any given time slot for $i,j=1...,N_H$.

An update packet can be transmitted using mode $j$ at harvester state $i$ only if the following constraint is satisfied:
\begin{equation}
 B(t)-P_j+P_{rec}(t)+P_{H_i}\geq 0,
 \label{eq:energy}
\end{equation}
where $P_{rec}(t)$ is the recovered energy at time $t$. Otherwise, the sensor either remains idle or uses another mode that satisfies the constraint. If the energy constraint is satisfied and an update packet is generated at time $t$ in state $H(t)=H_i$, the battery level is changed to $\left( B(t)-P_j+P_{rec}(t)+P_{H_i}\right)^-$ where we use the notation $B^-=\min(B,B_{max})$ \footnote{If the operator $(x)^-$ is used for the AoI, it denotes $\min(x,A_{max})$.}.

A simple scenario with $M=2$, $P_1=2$, $P_2=4$ is illustrated in Fig.~\ref{fig:sample_path}. In this example, we use an on-off energy harvesting process with either zero or two units of energy being harvested in a single time slot, depending on the state of the process. Moreover, we set $p_{rec}<1$, $N_{rec}=2$, which indicates that energy can be recovered during two idle time slots following a transmission. Upward (downward) arrows represent transmissions (receptions), whereas a cross indicates a transmission error. Transmissions with modes 1 and 2 are denoted by solid and dashed arrows, respectively. Note that the raises in the energy in intervals $t\in[0,2]$, $[6,9]$ are due to the harvested energy, whereas between $t\in[3,4]$, energy is recovered. On the other hand, between $t\in[8,9]$, both the harvested and recovered units of energy are stored in the battery. Finally, recovery opportunities are missed between $t\in[5,7]$, $[9,10]$ because of the probabilistic recovery model.
\begin{figure}[h!]
	\vspace{-6mm}
	\centering
	\includegraphics[width=5cm]{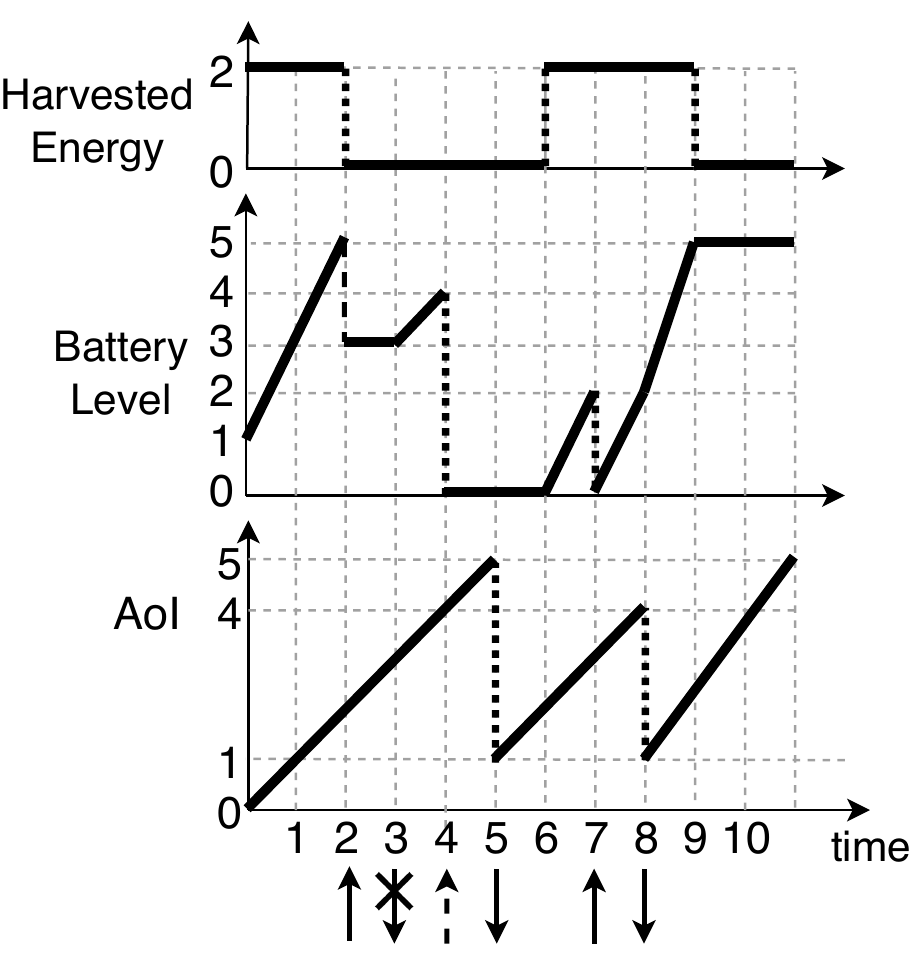}
	\caption{Sample path of the system for an on-off energy harvesting process and two transmission modes, represented by solid and dashed arrows.}
	\label{fig:sample_path}
\end{figure}
\vspace{-3mm}

\subsection{Markov Decision Process and the State Space}
We denote the system state at time $t$ by $S(t)$ which consists of the four-tuple $(A(t),T(t),H(t),B(t))$,
 where $T(t)\in\mathcal{T}=\{0,1,2,...,M,1^{(1)},...,1^{(N_{rec})},...,M^{(N_{rec})}\}$ denotes the system mode at time $t$. The first $M$ states in set $\mathcal{T}$ correspond to the TX modes $1,...,M$, whereas mode $i^{(j)}$ indicates that the sensor has been idle for $j$ consecutive time slots following a transmission in mode $i$, for $i=1,...,M$ and $j=1,...,N_{rec}$. Note that for the case with no battery recovery, $N_{rec}=0$. 

At the beginning of each time slot, the sensor selects a mode based on the state of the system. We denote the action taken by the sensor at time $t$ by $d(t)$. The set of actions that the sensor may take at any time slot is denoted by $\mathcal{D}$ and defined as $\mathcal{D}=\{0,1,...,M\}$, where $d(t)=m$ indicates the transmission of a new update packet with TX mode $m$, whereas the sensor remains idle if $d(t)=0$. We denote the transition probability from state $s$ to $s'$ when action $d$ is taken by $p_{ss'}(d)$, for $m\in\mathcal{D}$. Setting $p_{rec}=p\leq 1$, possible state transitions for $d(t)=0$ are given by $p_{ss'}(0) = $
\resizebox{0.89\linewidth}{!}{
	\begin{minipage}{\linewidth}
		\begin{eqnarray*}
			\begin{cases}
				q_{ij}p, & s = (a,m,H_i,b),\: s' = ((a+1)^-,m^{(1)},H_j,b')\\
				q_{ij}(1-p), & s = (a,m,H_i,b),\: s' = ((a+1)^-,m^{(1)},H_j,b'')\\
				q_{ij}p, & s = (a,m^{(n)},H_i,b),\: s' = ((a+1)^-,m^{(n+1)},H_j,b')\\
				q_{ij}(1-p), & s = (a,m^{(n)},H_i,b),\: s' = ((a+1)^-,m^{(n+1)},H_j,b'')\\
				q_{ij}, & s = (a,m^{(N_{rec})},H_i,b),\: s' = ((a+1)^-,0,H_j,b'')\\
				q_{ij}, & s = (a,0,H_i,b),\: s' = ((a+1)^-,0,H_j,b''),
			\end{cases}
		\end{eqnarray*}
	\end{minipage}
}
\vspace{2mm}

\noindent and similarly for $d(t)=m$, we have the transitions $p_{ss'}(m) = $
\resizebox{0.98\linewidth}{!}{
\begin{minipage}{\linewidth}
	\begin{eqnarray*}
		\begin{cases}
			q_{ij}(1-p_m), & s = (a,0,H_i,b),\: s' = (1,m,H_j,b_T)\\
			q_{ij}p_m, & s = (a,0,H_i,b),\: s' = ((a+1)^-,m,H_j,b_T)\\
			q_{ij}(1-p_m), & s = (a,m'^{(n')},H_i,b),\: s' = (1,m,H_j,b_T)\\
			q_{ij}p_m, & s = (a,m'^{(n')},H_i,b),\: s' = (1,m,H_j,b_T)\\
			q_{ij}(1-p_m), & s = (a,m',H_i,b),\: s' = (1,m,H_j,b_T)\\
			q_{ij}p_m, & s = (a,m',H_i,b),\: s' = ((a+1)^-,m,H_j,b_T)
		\end{cases}
	\end{eqnarray*}
\end{minipage}
}
\vspace{2mm}

\noindent where $m,m'=1,...,M$, $n=1,...,N_{rec}-1$, $n'=1,...,N_{rec}$, $i,j=1,...,N_H$, $b' =(b+P_{N_j}+P_m/N_{rec})^-$, $b'' =(b+P_{N_j})^-$, and $b_T=(b+P_{N_j}-P_m)^-\geq 0$.
\subsection{Performance Metrics for the AoI System}
We are interested in two particular performance metrics: i) the probability of the AoI hitting a predetermined upper-bound, denoted by $A_{max}$, and ii) the average AoI. Influenced by the term peak-age violation probability introduced in \cite{costa_2014}, the first performance metric is referred to as the peak-age hitting probability and is used to model the case where the information becomes useless (or less relevant to the system) after its age reaches $A_{max}$. This is in fact the case in several possible application areas of AoI; for instance, proximity sensors, or time-correlated data used for online learning with small time-windows. We use peak-age hitting probability to evaluate the worst-case performance, whereas we adopt the average AoI as the average performance indicator.
\vspace{-1mm}
\section{Optimal Update Packet Generation Policies}
In this section, we focus on obtaining \emph{stationary-deterministic} policies, which are defined as follows:
\begin{definition}
	A policy $\pi$ is stationary if the action taken at time $t$ and state $S(t)=s$ only depends on $s$, not the time index $t$. Similarly, a policy is called deterministic if it specifies the action at each state with certainty, i.e., with probability one.
\end{definition}
\subsection{Minimizing Peak-Age Hitting Probability}
We first characterize the optimal deterministic policy $\pi^{\ast}_p$ that minimizes the steady-state probability of the age hitting $A_{max}$, which is defined as $\pi^{\ast}_p(A_{max}) = \arg\min_{\pi\in \Pi} p_{\pi}(A_{max})$,
where $\Pi$ is the set of all possible policies and $p_{\pi}(A_{max})$ is the steady-state probability that $A(t)=A_{max}$ for the induced CTMC when policy $\pi$ is followed up to time $t$. Mathematically, we have $p_{\pi}(A_{max})=\lim\limits_{t\rightarrow \infty }\Pr(A(t)=A_{max}\lvert d_{0:t} \sim\pi)$ where $d_{0:t}$ is the sequence of decisions made up to time $t$. 


We use the following theorem to formulate the MDP where the objective is to maximize the long-run average reward:
\begin{theorem}
	\vspace{-1mm}
Minimizing the steady-state probability $p_{\pi}(A_{max})$ is equivalent to maximizing the average reward in the corresponding MDP with a negative reward associated with all states $s\in \mathcal{S}$ such that $s=(A_{max},m,H_i,b)$ for $m\in\mathcal{T}$, $H_i\in\{1,...,N_H\}$ and $b\in \{0,1,...,B_{max}\}$.
\end{theorem}
\begin{proof}
We denote the reward obtained at time $t$ by $R(t)$ and immediate reward associated with state $S(t)=s$ at time $t$ by $R_{s}$. We define $\mathcal{S}'\subset\mathcal{S}$ as the set of states having an AoI of $A_{max}$, each corresponding to a reward value of $r'<0$. We set the rewards for all the other states in the MDP to 0. The long-run average reward of the MDP in this case for a given policy $\pi$ is independent from the initial state and given by:\footnote{Note that since all the resulting Markov chains are ergodic, we represent the rewards by steady-state probabilities and use $v_p^{\pi}$ instead of $v_p^{\pi}(s)$. }
\vspace{-2mm}
\begin{align*}
v_p^{\pi}&=\liminf \limits_{T\rightarrow \infty} \mathbb{E}^{\pi}\left[ \frac{1}{T}\sum\limits_{t=1}^{T} R(t) \right]\\
&=\liminf \limits_{T\rightarrow \infty} \frac{1}{T}\sum\limits_{t=1}^{T}\mathbb{E}^{\pi}\left[  \sum_{s\in\mathcal{S}}R_{s}\Pr(S(t)=s\lvert d_{0:t-1}\sim \pi) \right]\\
&=\liminf \limits_{T\rightarrow \infty} \frac{1}{T}\sum\limits_{t=1}^{T}\mathbb{E}^{\pi}\left[  \sum_{s\in\mathcal{S}'}r'\Pr(S(t)=s\lvert d_{0:t-1}\sim \pi) \right]\\
&=\liminf \limits_{T\rightarrow \infty} \frac{1}{T}\sum\limits_{t=1}^{T}r'p_{\pi}(A_{max})=r'p_{\pi}(A_{max}).
\end{align*}
\end{proof}
\vspace{-3mm}
\subsection{Minimizing the Average Age}
We define $\mathcal{S}_k\subset \mathcal{S}$ as the set of states with AoI equal to $k$. Let $\bar{A}(\pi)$ denote the long-run average age of the system when policy $\pi$ is followed. We seek the optimal policy $\pi^{\ast}_a$ that minimizes $\bar{A}(\pi)$, defined as $\pi^{\ast}_a = \arg\min_{\pi\in \Pi} \bar{A}(\pi)$.
Each state $s\in\mathcal{S}_k$ has a reward of $-k$, which enables the use of a dual MDP problem as stated in the next theorem.
\begin{theorem}
Minimizing $\bar{A}(\pi)$ is equivalent to maximizing the long-run average reward of the MDP with reward $-k$ associated with $s\in\mathcal{S}_k$.
\label{thm:avg}
\end{theorem}
\begin{proof}
The long-run average reward for the defined MDP is
\resizebox{0.99\linewidth}{!}{
\begin{minipage}{\linewidth}
\begin{align*}
v_a^{\pi}&=\liminf \limits_{T\rightarrow \infty} \mathbb{E}^{\pi}\left[ \frac{1}{T}\sum\limits_{t=1}^{T} R(t) \right]\\
&=-\liminf \limits_{T\rightarrow \infty} \frac{1}{T}\sum\limits_{t=1}^{T}\mathbb{E}^{\pi}\left[  \sum_{k=1}^{\infty}\sum_{s\in\mathcal{S}_k}kPr(S(t)=s\lvert d_{0:t-1}\sim \pi) \right]\\
&=-\liminf \limits_{T\rightarrow \infty} \frac{1}{T}\sum\limits_{t=1}^{T}\sum_{k=1}^{\infty}kp_{\pi}(A_{k})=-\bar{A}(\pi).
\end{align*}
\end{minipage}}
\vspace{-1mm}
\end{proof}
In order to minimize the average age, it is required to solve an MDP with infinite state space due to unlimited age values. On the other hand, it is clear that the age value does not increase too much when a reasonably well performing policy is applied. This leads to using approximations of an AoI system with infinite state space using a finite state space MDP that upper bounds the age values, as in \cite{hsu_2017}. For this purpose, we use the iterative algorithm described in Algorithm~\ref{alg:finiteMDP} to find the value of $A_{max}$ for which the probability of hitting $A_{max}$ when policy $\pi^{\ast}_a$ is followed is negligible. Two input parameters to the algorithm are $K$, the maximum age value, and $\epsilon$, the threshold for the algorithm convergence. The algorithm computes the optimal $A_{max}$ for a given $\epsilon$, where the parameter $\epsilon$ trades off the computational complexity of the finite-MDP approximation against its accuracy.
\begin{algorithm}
\caption{Iterative finite-state MDP approximation}
\label{alg:finiteMDP}
\textbf{Initialize} $A_{max}=K$ for any $K$ such that $p_{\pi^{\ast}_a}(K)>\epsilon$ \\
 \While{$p_{\pi^{\ast}_a}(K)>\epsilon$}{
  Generate the state space with $A_{max}=K$\;
  Solve for $\pi^{\ast}_a = \arg\max_{\pi\in \Pi} v_a^{\pi}$\;
  \eIf{$p_{\pi^{\ast}_a}(K)>\epsilon$}{
   Increase $K$\;
   $A_{max}\gets K$\;
   }{
   Stop and record the optimal policy $\pi^{\ast}_a$ with $A_{max}=K$\;
  }
 }
\end{algorithm}
\vspace{-5mm}
\subsection{Minimizing the Weighted Sum of the Two Metrics }
In order to obtain a policy that takes into consideration both average and worst-case performances, we use a reward of $-\alpha k$ for the states with $A(t)=k$, $k=1,...,A_{max}-1$, and $-A_{max}$ when $A(t)=A_{max}$, and solve for the optimal policies that minimize the average reward. Note that $\alpha=0$ and $\alpha=1$ correspond to minimizing the peak-age hitting probability and average age, respectively. Thus, this new policy enables the tuning of the relative weight of average and worst-case performances. We refer to this policy as $\pi_{c}^{\ast}(\alpha)$.

Finally, we make the following observation that ensures the optimality of MDP problem formulations and their solutions:
\begin{proposition}
	For each metric, the resulting MDP problem has an average-optimal stationary deterministic solution.
\end{proposition}
\begin{proof}
 The proof is based on \cite[Ch.8]{puterman_2010}. Note that for both problems, the state space of the resulting Markov chain is finite, which suffices for a stationary-deterministic policy to exist.
 \vspace{-2mm}
\end{proof}

We use \emph{value-iteration algorithm} to obtain $\epsilon_c$-optimal solutions of the proposed MDP problems with a convergence parameter of $\epsilon_c=10^{-15}$, which performs quite well for our problem and is within an acceptable convergence time. Although we do not provide a formal proof for the convergence of the algorithm, we outline the main steps as follows. First, we notice that all feasible policies result in a unichain solution, with some states in the original MDP formulations possibly removed because of the particular policy. Moreover, all resulting Markov chains are aperiodic. To see this, consider an energy harvesting process with more than one state which either transits into another state or remains in the same state at a given time slot. This means that after leaving any state, that state can be traversed again both after $n$ and $n+1$ time slots for some $n> 1$. Hence, the resulting Markov chain for all energy profiles of interest have a period of one, i.e., are aperiodic. For unichain and aperiodic MDPs, convergence of the value-iteration algorithm in finite number of steps is proved in \cite[Ch.8]{puterman_2010}.
\section{Numerical Examples}
We set $A_{max}=20$ when we compute $\pi^{\ast}_p(A_{max})$ unless otherwise stated. For the average age, we run Algorithm~\ref{alg:finiteMDP} for each case separately to obtain the optimal $A_{max}$. We consider an on-off energy harvesting process and set $q_{12}=q_{21}=0.1$ and $q_{11}=q_{22}=0.9$.

In the first example, we investigate how the system benefits from the battery recovery effect and having two TX modes. For this purpose, we consider two TX modes that comply with MCS-5 index in LTE \cite{ikuno_2010}. TX mode 1 corresponds to a low-power transmission scheme with $P_1=3$ and a higher error probability $p_1=0.4$, whereas the high-power TX mode with $P_2=6$, mode 2, has a lower probability of error of $p_2=10^{-3}$. The first observation we make from Figures~\ref{fig:ex1age}(a) and \ref{fig:ex1prob}(a) is that incorporating the battery recovery effect significantly improves the system performance in terms of both metrics. Moreover, having both low and high power TX modes aids in reducing both performance metrics significantly, especially when the battery recovery effect is also considered. Finally, Figures~\ref{fig:ex1age}(b) and \ref{fig:ex1prob}(b) depict that the optimal policy that minimizes one metric may actually perform quite poorly in terms of the other one. For instance, for the case with battery recovery and both TX modes available, when the average age is minimized with policy $\pi_a^*$, the resulting peak-age hitting probability is around $4\times10^{-4}$, which can be reduced to $10^{-10}$ when policy $\pi_p^*$ is followed. 
\begin{figure}[h!]
	\centering
	\includegraphics[width=8.7cm]{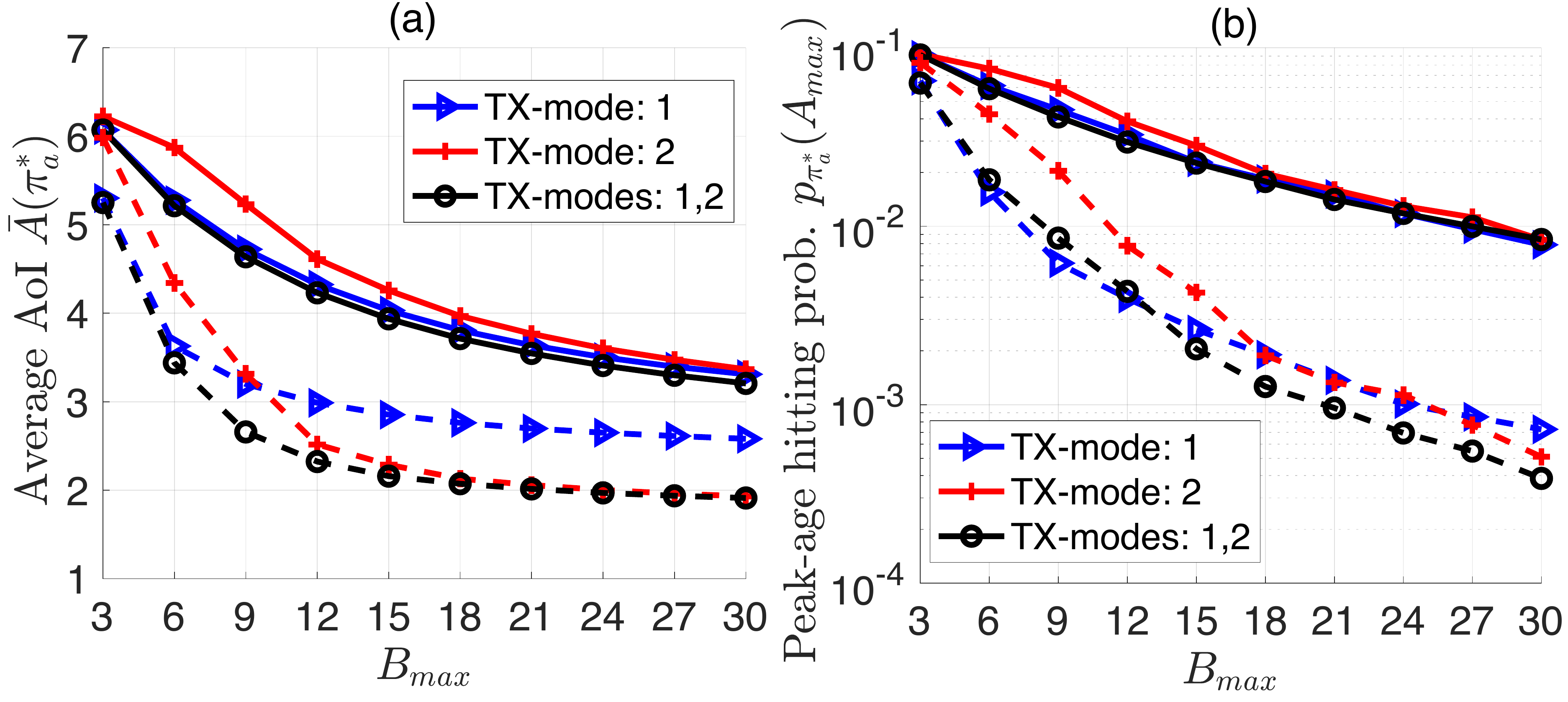}
	\caption{(a) Optimal policies with minimum average AoI and (b) corresponding peak-age hitting probabilities, as functions of $B_{max}$. Dashed (solid) lines corresponds to the scenarios with (without) recovery effect.}
	\label{fig:ex1age}
\end{figure}
\begin{figure}[h!]
	\centering
	\includegraphics[width=8.7cm]{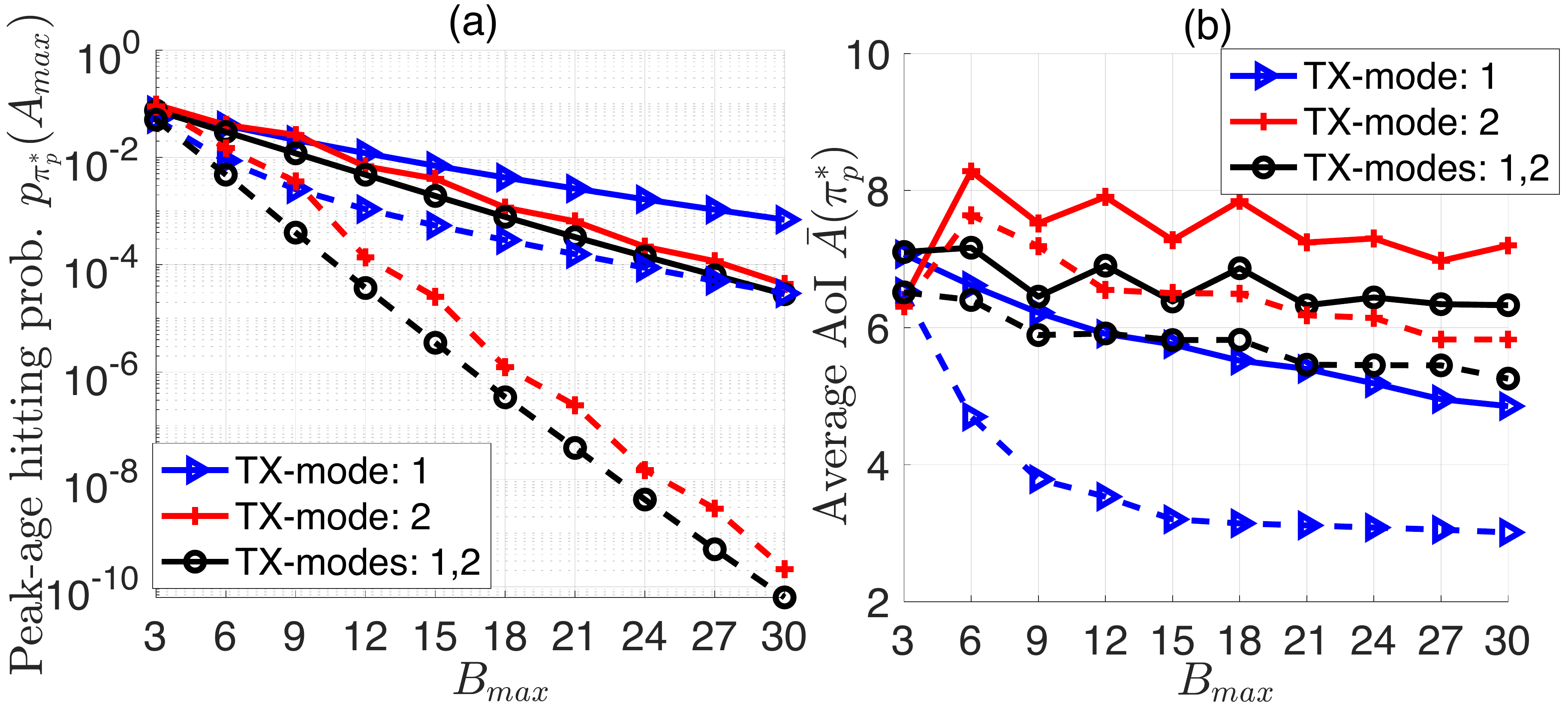}
	\caption{(a) Optimal policies with minimum peak-age hitting probabilities and (b) corresponding average AoI, as functions of $B_{max}$.  Dashed (solid) lines corresponds to the scenarios with (without) recovery effect.}
	\label{fig:ex1prob}
	\vspace{-2mm}
\end{figure}

In the second example, we investigate the performance of joint-optimal policy $\pi_c^*(\alpha)$ for varying values of $\alpha$ and three different values of battery recovery probability $p_{rec}$. We illustrate the average age and peak-age hitting probability achieved with the joint-optimal policy $\pi_c^*(\alpha)$ in Figures~\ref{fig:ex2comb}(a) and \ref{fig:ex2comb}(b), respectively, which reveal how average and worst-case performances can be traded off by changing $\alpha$. We also observe from these two figures that recovering the consumed energy with higher probability benefits the system more, as expected. Finally, from Figures~\ref{fig:ex2comb}(c) and \ref{fig:ex2comb}(d), we observe that the system tends to operate at higher average TX power ($\bar{P}_T$) and lower average battery level ($\bar{B}$) in general for increased $\alpha$. This is mainly because when more weight is given to minimizing the average age, the sensor does not wait until the battery level, and hence the age, increases too much and also favors the high power (or low error) mode to reduce the age rapidly, especially for higher $p_{rec}$. However, this behavior also depends on the set of parameters, including error probability and power of each TX mode, and we may observe a different trend for a different setting.
\begin{figure}[h!]
	\centering
	\includegraphics[width=8.6cm]{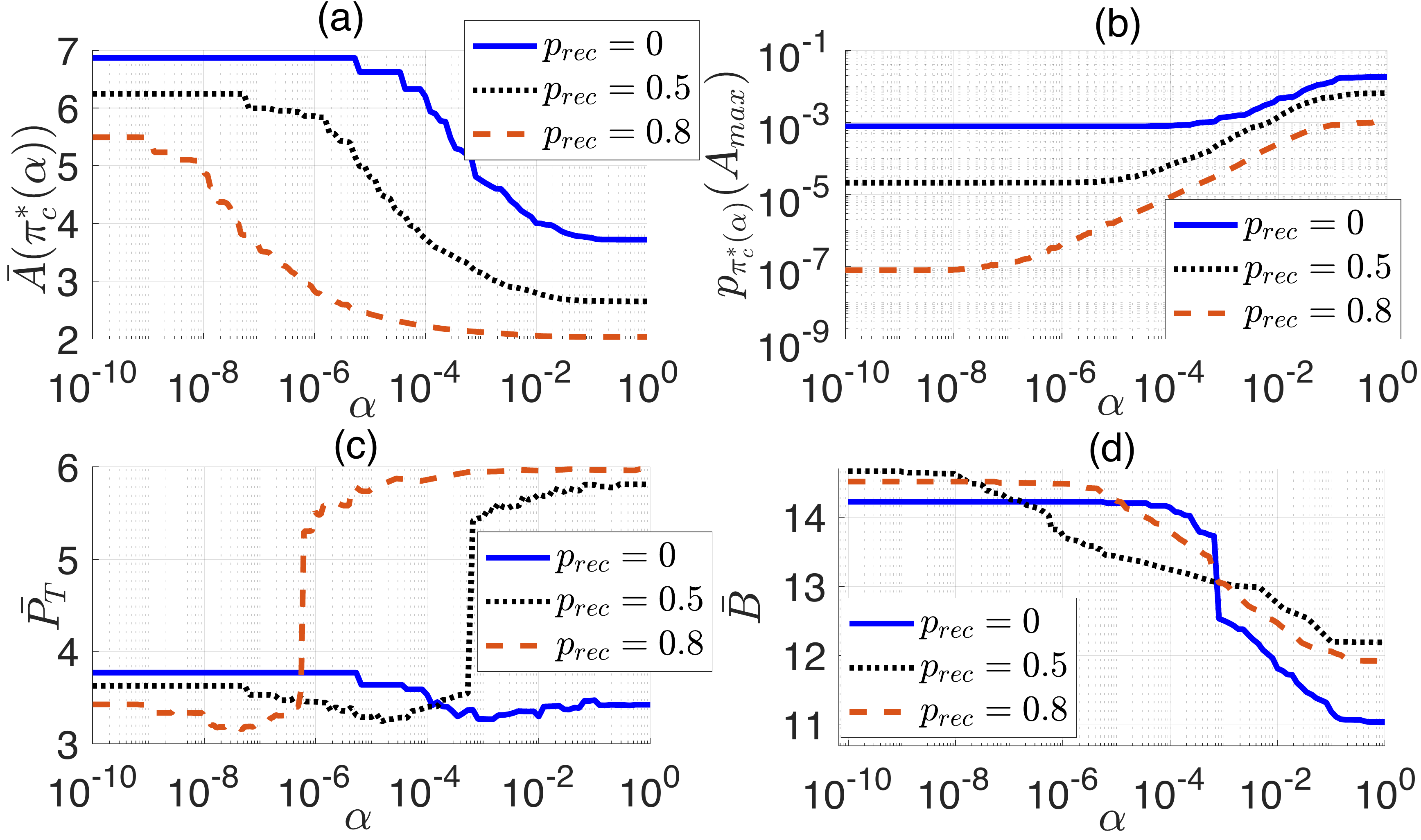}
	\caption{(a) Average age, (b) peak-age hitting probability, (c) average TX power and (d) average battery level achieved by the joint-optimal policy $\pi_c^*(\alpha)$, as functions of $\alpha$, for three different battery recovery probability.}
	\label{fig:ex2comb}
	\vspace{-1mm}
\end{figure}
\section{Conclusion}
We investigated the average and worst-case optimal policies for an energy harvesting age of information (AoI) system, where a sensor takes advantage of multiple available transmission (TX) modes and probabilistic battery recovery. After obtaining the optimal policies in terms of average age or peak-age hitting probability, which quantify the performance in terms of the average and worst-case performances, respectively, as solutions of corresponding Markov Decision Process formulations, we first showed the benefits of incorporating multiple transmission modes and the battery recovery effect. We then proposed a joint formulation that considers both performance metrics together and enables the system to operate at any point between average-optimal and worst-case-optimal policies.
 \bibliographystyle{ieeetr}
 \bibliography{whole}
\end{document}